\documentclass{sig-alternate}
\usepackage{array}
\usepackage{multirow}
\usepackage[ruled,vlined,titlenumbered]{algorithm2e}

\newtheorem{definition}{Definition}
\newtheorem{lemma}{Lemma}
\newtheorem{example}{Example}
\newtheorem{proposition}{Proposition}

\begin{document}
%
\conferenceinfo{WOODSTOCK}{'97 El Paso, Texas USA}

\title{Towards a semantic and statistical selection of association rules}
%
%
%
%
%

\numberofauthors{4} 
%

\author{
%
%
\hspace*{-4cm}
\alignauthor
Slim Bouker\\
       \affaddr{Clermont university, Blaise Pascal University, LIMOS, BP 10448, F-63000 Clermont-Ferrand, France}\\
       \affaddr{CNRS, UMR 6158, LIMOS, F-63173 Aubi\`ere, France}\\
       \email{bouker@isima.fr}
\alignauthor
Rabie Saidi\\
 \affaddr{European Bioinformatics Institute}\\
       \affaddr{Hinxton, Cambridge, CB10 1SD, United Kingdom}\\
       \email{rsaidi@ebi.ac.uk}
\alignauthor
Sadok Ben Yahia\\
        \affaddr{University of Sciences of Tunis, Department of computer science, 1060 Tunis, Tunisia}\\
       \email{sadok.benyahia@fst.rnu.tn}
\and
\alignauthor
Engelbert Mephu Nguifo\\
       \affaddr{Clermont university, Blaise Pascal University, LIMOS, BP 10448, F-63000 Clermont-Ferrand, France}\\
       \affaddr{CNRS, UMR 6158, LIMOS, F-63173 Aubi\`ere, France}\\
       \email{mephu@isima.fr}
}

\maketitle
\begin{abstract}
The increasing growth of databases raises an urgent need for more accurate methods to better understand the stored data. In this scope, association rules were extensively used for the analysis and the comprehension of huge amounts of data. However, the number of generated rules is too large to be efficiently analyzed and explored in any further process. Association rules selection is a classical topic to address this issue, yet, new innovated approaches are required in order to provide help to decision makers. Hence, many interestingness measures have been defined to statistically evaluate and filter the association rules. However, these measures present two major problems. On the one hand, they do not allow eliminating irrelevant rules, on the other hand, their abundance leads to the heterogeneity of the evaluation results which leads to confusion in decision making. In this paper, we propose a two-winged approach to select statistically interesting and semantically incomparable rules. Our statistical selection helps discovering interesting association rules without favoring or excluding any measure. The semantic comparability helps to decide if the considered association rules are semantically related \emph{i.e} comparable. The outcomes of our experiments on real datasets show promising results in terms of reduction in the number of rules.
\end{abstract}





\vspace{-0.3cm}
\section{Introduction}

Mining association rules is one of the core tasks in data mining. Since its first formalization in \cite{Agra93}, the association rules research field has become very popular. Indeed, mining association rules provides an opportunity to extract relevant and valuable relationships between attributes in transaction databases. Currently, association rules are widely used in various areas of \emph{decision making} such as communication networks, market and risk management, inventory control, etc.
However, existing association rules algorithms produce an overwhelming number of rules \cite{Klemettinen94, Mann97}.
Hence, the decision maker is unable to determine the most interesting ones and is consequently unable to make decisions. In order to overcome this shortcoming, an efficient evaluation of rules has become a compelling need rather than being a rational choice. Several works have been devoted to the study of the interestingness of association rules \cite{Hilderman, Hilderman2, tan, Vail04}. As a consequence, a panoply of statistical measures have been proposed allowing the evaluation of rules from different sights. Although, the abundance of these measures ($\approx 60$) has raised another problem for the decision maker. In fact, a given rule considered relevant according to one measure may be irrelevant with respect to another one. Hence, the output of evaluation vary from one measure to another and may even be contradictory. This has led to a trend of works that focus on proposing approaches to assist the user in selecting the measures that best fit the decision scope.

In this scope, existing approaches can be classified into two main categories namely the expert-based approaches and the property-based ones. In the first category, different studies compared the ranking of rules by human experts to that yield by various measures. Then, they suggested choosing the measure that yields the closest one to the expert ranking \cite{Ohsaki, Tan02}. The results issued from these studies are highly related to specific datasets and experts. Thus, they cannot be taken as general conclusions. Moreover, in a real problem, it is not always possible to easily get an expert's ranking. As for the second category, the selection of measures is based on many properties reported in \cite{geng07}. Using properties facilitates a general and practical way to automatically identify interesting measures. Geng and Hamilton surveyed the interestingness of measures and summarized nine properties. This trend has been enriched by different other works \cite{Blan05, heravi, lenca08, mad} with an additional number of properties. Nevertheless, these properties are not standards \cite{lenca08}, since they do not guarantee selecting only one best measure. Indeed, a wide range of UCI\footnote{http://archive.ics.uci.edu/ml/} datasets were used to study the impact of different properties. The results show that no single measure can be elected as an obvious winner \cite{heravi}. Then, in the case of selecting many measures, the problem related to the variety of outputs, mentioned above, persists. In other words, the user cannot proceed towards a unique selection of rules.


Our contribution lies within this scope. In this paper, we introduce a novel approach that aims to discover interesting association rules without favoring or excluding any measure among the used measures. For this purpose, we integrate into the rule selection process, the \emph{skyline operator} \cite{Bor} whose fundamental principle relies on the notion of \emph{dominance}. The skyline operator is used to resolve mathematical and economics problems such as maximum vectors \cite{Kung}, Pareto set \cite{Mat} and multi-objective optimization \cite{Ste}. Besides, the skyline operator has received considerable attention in database community and several algorithms based on block nested loops \cite{Bor}, divide-and-conquer search \cite{Kossmann02} and index scanning \cite{Tan2001}, have been developed to meet skyline constraints in various computational domains. In our work, we use the skyline operator to detect the most interesting rules when considering several measures. The dominance relationship, which is the corner stone of the skyline operator, is applied on rules and can be presented as follows: a rule $r$ is said \emph{dominated} by another one $r'$, if for all used measures, $r$ is less relevant than $r'$. The former rule (\emph{i.e.,} $r$) is discarded from the result, not because it is not relevant for one of the measures, but because it is not interesting according to the combination of all measures. Even though the dominance relationship allows discovering interesting rules with respect to all the measures, it does not consider the semantic relationship between rules. In real world applications not all the rules are comparable, since different rules may belong to different semantic context. Hence, it would not be judicious that an undominated rule eliminates another rule while they are semantically independent. In our approach, we considered the semantic relationship constraints during the selection of undominated rules. It is worth mentioning that our method bypasses another non-trivial problem which is the threshold value specification.

The remainder of this paper is organized as follows. Section \ref{s2} gives brief definitions related to association rules and introduces the dominance relationship. We propose and detail our approach of rule selection in section \ref{s3}. Results of the experiments carried out on several datasets are reported in section \ref{s4}. Concluding points and avenues of future work are sketched in section \ref{s5}.


\section{Preliminary definitions and problem formulation} \label{s2}
In this section, we first recall basic definitions related to association rules. Then, we present these rules as numeric vectors within the same dimension after having been statistically evaluated by a set of measures. These vectors allow us to benefit from the concept of dominance and adapt it to select interesting rules.

\subsection{Association rules and interestingness measures}
Let $\mathcal{I}$ be a set of literals called \emph{items}, an itemset corresponds to a non null subset of $\mathcal{I}$.
These itemsets are gathered together in the set $\mathcal{L}$ : $\mathcal{L}$ = 2$^{\mathcal{I}}$$\setminus$$\emptyset$.
In a transactional dataset, each transaction contains an itemset of $\mathcal{L}$. Table \ref{example}(a) sketches a transactional dataset $\mathcal{D}$
where 10 transactions, denoted by $t_{1}$, . . . , $t_{10}$  described by 4 items denoted by $a$, $b$, $c$, $d$. The support of an itemset X, denoted $supp$(X), is the number of transactions containing $X$.

An association rule $r$ is a relation between itemsets of the form $r$: X$\rightarrow$Y where $X$ and $Y$ are itemsets,
and $X$$\cap$$Y$$=$$\emptyset$. Itemsets $X$ and $Y$ are called, respectively, premise and conclusion of $r$. The support of $r$ is
equal to the number of transactions containing both $X$ and $Y$, $supp$($r$)= $supp$(X$\cup$Y). As defined in \cite{Agra93}, given a typical
market-basket database, the association rule $r$: X$\rightarrow$Y means if someone buys the itemset X then he probably also buys Y.
The statical interesting of an association rules is evaluated using measures that are usually expressed as a function of support counts as presented in Table \ref{example}(c).


\begin{table}[htpd]


\begin{minipage}[t]{.1\linewidth}
    \begin{tabular}{|c|c c c c|}
  \hline
   & $a$ & $b$ & $c$ & $d$ \\\hline
   $t_{1}$& & &$\times$  & $\times$    \\
   $t_{2}$&$\times$  &  &  &  \\
   $t_{3}$& $\times$ &  &  & $\times$ \\
   $t_{4}$&  &  & $\times$ &  \\
   $t_{5}$&  & $\times$ &  & $\times$ \\
   $t_{6}$& $\times$ &  &  & $\times$ \\
   $t_{7}$&  &  & $\times$ &  \\
   $t_{8}$&  &  &  & $\times$ \\
   $t_{9}$&  & $\times$ & $\times$ & $\times$ \\
   $t_{10}$&  & $\times$ & $\times$ &  \\\hline
\end{tabular}
\hspace*{-0.18cm}$(a)\ A\ transaction\ dataset\ \mathcal{D}$
\end{minipage}
\hfill
\begin{minipage}[t]{0.55\linewidth}
    \begin{tabular}{|c||c c c|}
  \hline
       $Rule$    & $Freq$   & $Conf$   & $Pearl$  \\ \hline
    $r_{1}$: $a$$\rightarrow$$d$ & 0.20 & 0.66 & 0.02 \\
    $r_{2}$: $b$$\rightarrow$$c$ & 0.20 & 0.66 & 0.05 \\
    $r_{3}$: $b$$\rightarrow$$d$ & 0.20 & 0.66 & 0.02 \\
    $r_{4}$: $c$$\rightarrow$$b$ & 0.20 & 0.40 & 0.05 \\
    $r_{5}$: $c$$\rightarrow$$d$ & 0.20 & 0.40 & 0.10 \\
    $r_{6}$: $d$$\rightarrow$$a$ & 0.20 & 0.33 & 0.02 \\
    $r_{7}$: $d$$\rightarrow$$b$ & 0.20 & 0.33 & 0.01 \\
    $r_{8}$: $d$$\rightarrow$$c$ & 0.20 & 0.33 & 0.10 \\
    $r_{9}$: $b$$\rightarrow$$cd$ & 0.10 & 0.33 & 0.03 \\
    $r_{10}$: $c$$\rightarrow$$bd$ & 0.10 & 0.20 & 0.00 \\
    $r_{11}$: $d$$\rightarrow$$bc$ & 0.10 & 0.16 & 0.02 \\
    $r_{12}$: $bc$$\rightarrow$$d$ & 0.10 & 0.50 & 0.02 \\
    $r_{13}$: $bd$$\rightarrow$$c$ & 0.10 & 0.50 & 0.00 \\
    $r_{14}$: $cd$$\rightarrow$$b$ & 0.10 & 0.50 & 0.04 \\\hline
\end{tabular}
\hspace*{0.7cm}$(b)\ A\ table\ relation\  \Omega(\mathcal{R},\mathcal{M})$
\end{minipage}

\vspace*{0.7cm}
  \begin{minipage}[t]{.4\linewidth}
    \begin{tabular}{|c|c|c|}
  \hline
   $\textbf{Name}$& $\textbf{Definition}$ & $\textbf{Domain}$ \\ \hline
  $Frequency$& $\frac{supp(X\cup Y)}{|D|}$ & [0, 1] \\
        & &\\
   $Confidence$& $\frac{supp(X\cup Y)}{supp(X)}$ & [0, 1] \\
   & &\\
   $Pearl$& $\frac{supp(X)}{|D|}$ $\times$ $\mid$$\frac{supp(X\cup Y)}{supp(X)}$ $-$ $\frac{supp(Y)}{|D|}$$|$ & [0, 1]\\\hline
\end{tabular}
\hspace*{2.4cm}$(c)\ Some\ measures\ of\ \mathcal{M}$
\end{minipage}

\caption{Example of a dataset transaction and measures.}
\label{example}
\end{table}


\vspace*{0.4cm}
\subsection{Undominated association rules}

After mining association rules from a transactional dataset $\mathcal{D}$ (\emph{e.g.,} Table \ref{example}(a)), a set $\mathcal{R}$ of rules is obtained (\emph{e.g.,} Table \ref{example}(b) first column). Rules of $\mathcal{R}$ are evaluated with respect to a set $\mathcal{M}$ of measures (\emph{e.g.,} Table \ref{example}(c)) to form a relational table $\Omega$ (\emph{e.g.,} Table \ref{example}(b)). Formally, $\Omega$ = ($\mathcal{R}$,$\mathcal{M}$) with the set $\mathcal{M}$ = $\{$$m_{1}$,$\ .\ .\ .$, $m_{k}$$\}$ of measures as attributes and the set $\mathcal{R}$ = $\{$$r_{1}$,$\ .\ .\ .$, $r_{n}$$\}$ of rules as objects. We denote by $r$[$m$] the value of the measure $m$ for the rule $r$, $r$ $\in$ $\mathcal{R}$ and $m$ $\in$ $\mathcal{M}$. Since the evaluation of rules varies from a measure to another one, using several measures could lead to different outputs (relevant rules with respect to a measure). For example, $r_{1}$, $r_{2}$ and $r_{3}$ are the best rules with respect to the \emph{Confidence} measure whereas it is not the case according to the evaluation of \emph{Pearl} measure which favors $r_{5}$. This difference of evaluations is confusing for any process of rule selection.

Based on the above formulation of $\Omega$, we can utilize the notion of dominance between rules to address the selection of relevant ones. Before, formulating the dominance relationship between rules we need to define it at the level of measure values. To do that, we define value dominance as follows:

\begin{definition}(Value Dominance)
Given two values of a measure $m$ corresponding to two rules $r$ and $r'$, we say that $r$[$m$] dominates $r'$[$m$],
denoted by $r$[$m$]\ $\succeq$\ $r'$[$m$], iff $r$[$m$] is preferred to $r'$[$m$].
If $r$[$m$]\ $\succeq$\ $r'$[$m$] and $r$[$m$]\ $\neq$\ $r'$[$m$] then we say that $r$[$m$] strictly dominates $r'$[$m$], denoted
$r$[$m$]\ $\succ$\ $r'$[$m$].
\end{definition}

\vspace{-0.4cm}

\begin{definition}(Rule Dominance)
Given two rules $r$, $r'$ $\in$ $\mathcal{R}$, the dominance relationship according to the set of measures $\mathcal{M}$ is defined as follows:
\begin{itemize}
  \item[-] $r$ dominates $r'$, denoted $r$ $\succeq$ $r'$,\ \ iff\ \ $r$[$m$] $\succeq$ $r'$[$m$], $\forall$ $m$ $\in$ $\mathcal{M}$.\vspace*{-0.2cm}
  \item[-] If $r$ $\succeq$ $r'$ and $r'$ $\succeq$ $r$, \emph{i.e.,} $r$[$m$] $=$ $r'$[$m$], $\forall$ $m$ $\in$ $\mathcal{M}$ then $r$ and $r'$ are said \textbf{equivalent}, denoted $r$ $\equiv$ $r'$.\vspace*{-0.2cm}
  \item[-] If $r$ $\succeq$ $r'$ and $\exists$ $m$ $\in$ $\mathcal{M}$ such that $r'$[$m$] $\succ$ $r$[$m$] , then $r'$ is \textbf{strictly dominated} by $r$ and we note $r$ $\succ$ $r'$.
\end{itemize}
\end{definition}
It is easy to check that the strict dominance relationship fulfils the following properties:
\begin{itemize}
  \item[-]\textbf{irreflexive}: $r$ $\not\succ$ $r$,\ \textit{i.e,} $r$ $\succ$ $r$ is false for each $m$ $\in$ $\mathcal{M}$,\vspace*{-0.2cm}
  \item[-]\textbf{transitive}: $\forall$ $r$, $r'$ and $r''$ $\in$ $\mathcal{R}$, if $r$ $\succeq$ $r'$ and $r'$ $\succeq$ $r''$ then $r$ $\succeq$ $r''$.
\end{itemize}

\begin{example} Given the relation table $\Omega$ in Table \ref{example}(b), the rule $r_{2}$ strictly dominates $r_{1}$ since $r_{2}$[$Freq$] $\succeq$ $r_{1}$[$Freq$], $r_{2}$[$Conf$] $\succeq$ $r_{1}$[$Conf$] and $r_{2}$[$Pearl$] $\succ$ $r_{1}$[$Pearl$].
\end{example}
Whenever a rule $r$ dominates another one $r'$ with respect to $\mathcal{M}$, this means that $r$ is equivalent to or better than $r'$ for all measures. Hence, the dominance relationship allows comparing concurrently two rules with respect to all measures. Hence, it can be used to bypass the problem of difference of evaluations. Rules dominated by other ones (at least one), according to $\mathcal{M}$, are not relevant and have to be eliminated. The skyline operator for association rules formalizes this intuition.

\vspace*{-0.2cm}
\begin{definition}(Skyline operator)
The skyline of $\Omega$ over $\mathcal{M}$, denoted by $Sky_{M}$($\Omega$), is the set of rules from $\Omega$ defined as follows:
\vspace*{-0.3cm}
\begin{center}
$Sky_{\mathcal{M}}$($\Omega$) $=$ $\{$ $r$$\in$ $\mathcal{R}$ $\mid$ $\not\exists$ $r'$ $\in$ $\mathcal{R}$, $r'$ $\succ$ $r$$\}$
\end{center}
\end{definition}

In other words, the skyline of $\Omega$ is the set of undominated rules of $\mathcal{R}$ with respect to $\mathcal{M}$. For instance, from the relation table $\Omega$ in Table \ref{example}(b), $Sky_{M}$($\Omega$) $=$ $\{$$r_{2}$, $r_{5}$$\}$ since there is no rule in $\mathcal{R}$ dominating $r_{2}$ or $r_{5}$.


\subsection{Comparable association rules}

Mining the set of undominated rules allows eliminating irrelevant ones. Precisely, each undominated rule in $Sky_{M}$($\Omega$) removes all the rules it dominates. However, in real world applications not all the rules are comparable, since different rules may belong to different semantic context. Hence, it would not be judicious that an undominated rule eliminates another rule while they are semantically independent. Therefore, the dominance should not be the only criteria to define the rules to keep and those to eliminate. Another criterion must be introduced to ensure some semantic side in the selection process. This criterion would define a kind of semantic relationship between rules and restrict the use of dominance. Concretely, the dominance between two rules must be applied only if a semantic relationship exists between them. For this purpose, we define a semantic relationship called comparability.

\begin{definition}(comparability)
Two rules $r$: X$\rightarrow$Y and $r'$: $X'$$\rightarrow$$Y'$ are said comparable, we note $comp$($r$, $r'$) = true iff
($X$ $\subseteq$ $X'$ and $Y$ $\subseteq$ $Y'$) or ($X'$ $\subseteq$ $X$ and $Y'$ $\subseteq$ $Y$).
\end{definition}
For instance, from the relation table $\Omega$ in Table \ref{example}(b), we have $r_{1}$: $a$$\rightarrow$$d$ only dominated by $r_{2}$: $b$$\rightarrow$$c$ but the two rules are not comparable. Hence, $r_{1}$ should not be discarded. It is easy to check that the comparability relationship fulfils the following properties:
\begin{itemize}
  \item[-]\textbf{reflexive}: $\forall$ $r$ $\in$ $\mathcal{R}$, $comp$($r$, $r$) = $true$\vspace*{-0.2cm}
  \item[-]\textbf{non-transitive}: $\exists$ $r$, $r'$ and $r''$ $\in$ $\mathcal{R}$ such that $comp(r, r') = true$ and $comp$($r'$, $r''$) = $true$ but $comp$($r$, $r''$) = $false$.
\end{itemize}

\begin{definition}\label{comp}
Let $r$ and $r'$ be two rules. We said $r'$ is incomparable with $r$ iff $r$ $\succ$ $r'$ and comp($r$,$r'$)= false\\
All rules incomparable with $r$ are denoted by $Icomp$($r$); \vspace*{-0.2cm}
\begin{center}
$Icomp$($r$) $=$ $\{$ $r'$$\in$ $\mathcal{R}$ $\mid$  $r$ $\succ$ $r'$ $\wedge$ comp($r$,$r'$)= false$\}$
\end{center}
\end{definition}

The motivation behind the concept of comparability resides in the fact that, with reference to a given rule, some additional or missing information would yield a new rule with better or worse statistical interestingness. This amount of information, that we call sematic differential, could be simply additional/missing items in the premise of a rule and/or in its conclusion.
Given two comparable rules, we can make a one-way reading of the semantic differential from one rule to the other. As the syntax and the semantics change between the two rules, we can notice 2 cases:
\vspace*{-0.1cm}
\begin{enumerate}
  \item No rule is dominating the other. Hence, both of them are kept.\vspace*{-0.2cm}
  \item One of the two rules dominates the other. In this case, it would be suitable to remove the dominated rule as long as the dominant rule is not dominated by a third rule.
\end{enumerate}

The comparability relationship, we have defined, is one way to express semantics between rules. Obviously, there may exist several other ways to reveal different semantics \cite{Ping}, \cite{Roddick}. Generally, in this context inferring semantics between rules relies on their syntax comparison. For instances, semantically related rules may have a common itemset or a common premise or a common conclusion, etc.

\section{Representative association rules}\label{s3}

By mixing together the concepts of dominance and comparability, we propose a selection method that output inter-independent and statistically relevant rules. We call them representative rules.


\begin{definition}(Representative rules)
The representative association rules of $\Omega$ over $\mathcal{M}$, denoted by $\mathcal{RR}$, is the set of rules from $\Omega$ defined as follows:\\
$\mathcal{RR}$$_{\mathcal{M}}$($\Omega$) $=$ $\{$ $r$ $\in$ $\mathcal{R}$ $\mid$ $\not\exists$ an undominated rule $r'$, $r'$ $\succ$ $r$ $\wedge$ comp($r$,$r'$)= $true$$\}$
\end{definition}


\begin{proposition}
The following property holds:\\
$Sky_{\mathcal{M}}$($\Omega$) $\subseteq$ $\mathcal{RR}$$_{\mathcal{M}}$($\Omega$)
\end{proposition}

Hence, any undominated rule is a representative rule.

\subsection{$\mathcal{RR}$ construction}

To discover the representative association rules, a naive approach consist in comparing each rule with all other ones. However, association rules are often present in huge number which makes pairwise comparisons costly. In the following, we show how to overcome this problem by adopting the principle of approaches oriented divide-and-conquer search \cite{Kossmann02} used for answering queries in database applications. First, we introduce the
notion of reference rule.
\begin{definition}(Reference Rule)
A reference rule $r^{\perp}$ is a fictitious rule that dominates all the rules of $\mathcal{R}$. Formally:
$\forall$ $r$ $\in$ $\mathcal{R}$, $r^{\perp}$$\succeq$$r$.
\end{definition}
\begin{example} From the relational table $\Omega$ given in Table \ref{example}, we can consider $r^{\perp}$ as the fictitious rule such that for each measure $m$ $\in$ $\mathcal{M}$, $r^{\perp}[m]$ is the maximal value appearing in the active domain of $m$, \textit{i.e.,} $r^{\perp}$ $=$ $\langle$0.2, 0.66, 0.10$\rangle$. Hence, it does not exist any rule in $\mathcal{R}$ that dominates $r^{\perp}$.
\end{example}
In practice, measures are heterogenous and defined within different domains. For our purpose, $\mathcal{M}$ has to be normalized into $\widehat{\mathcal{M}}$ within one interval [$p$,$q$]. In other words, each measure $m$ $\in$ $\mathcal{M}$ must be normalized into $\widehat{m}$ $\in$ $\widehat{\mathcal{M}}$ within [$p$,$q$]. The normalization of a given measure $m$ is performed depending on its domain and the statistical distribution of its active domain. We recall that the active domain of a measure $m$ is the set of its values in $\Omega$. The normalization is a statistical problem which is beyond the scope of this paper.
It is worth mentioning, the normalization of a measure does not modify the domination relationship between two given values.

\begin{definition}(Degree of similarity)
Given two rules $r$, $r'$ $\in$ $\mathcal{R}$, the degree of similarity between $r$ and $r'$ with respect to $\widehat{\mathcal{M}}$
is defined as follows:

 \begin{displaymath}
 DegSim(r,r') = \frac{\sum_{i=1}^{k}\mid r[\widehat{m}_{i} ] - r'[\widehat{m}_{i} ]\mid}{k}
 \end{displaymath}

with $\mid x - y \mid$ is the absolute value of $(x-y)$, x and y $\in$ $[$$p$,$q$$]$ and $k$ $=$ $\mid \widehat{\mathcal{M}} \mid$.
\end{definition}

\begin{example} Let's consider our running example using the relation table $\Omega$ in Table \ref{example}(b). Since all measures are defined within the same domain [0,1], we can compute, without normalization, the degree of similarity between each rule and the reference rule given in the previous example. $DegSim$ ($r^{\perp}$,$r_{1}$) = 0.08, $DegSim$($r^{\perp}$,$r_{2}$) = 0.01, $DegSim$($r^{\perp}$,$r_{3}$) = 0.08, $DegSim$($r^{\perp}$,$r_{4}$) = 0.10, $DegSim$($r^{\perp}$,$r_{5}$) = 0.08, $DegSim$($r^{\perp}$,$r_{6}$) = 0.13, $DegSim$($r^{\perp}$,$r_{7}$) = 0.14, $DegSim$ ($r^{\perp}$,$r_{8}$) = 0.11, $DegSim$($r^{\perp}$,$r_{9}$) = 0.20, $DegSim$($r^{\perp}$,$r_{10}$) = 0.22, $DegSim$($r^{\perp}$,$r_{11}$) = 0.22, $DegSim$($r^{\perp}$,$r_{12}$) = 0.11, $DegSim$($r^{\perp}$,$r_{13}$) = 0.08, $DegSim$($r^{\perp}$,$r_{14}$) = 0.10.
\end{example}

After giving the necessary definitions (reference rule and degree of similarity), the following lemma gives a remedy to
the issue evoked in the beginning of section 3.1. Indeed, it offers a swifter solution rather than pairwise comparisons; to
find representative rules.

\begin{lemma}\label{l1}
Let $r$ $\in$ $\mathcal{R}$ be a rule having the minimal degree of similarity with respect to $r^{\perp}$, then $r$ $\in$ $\mathcal{RR}$$_{\mathcal{M}}$($\Omega$).
\end{lemma}

\begin{proof}
Let $r$ $\in$ $\mathcal{R}$ be a rule having the minimal degree of similarity with respect to $r^{\perp}$ and we suppose that $r$ $\not\in$ $\mathcal{RR}$$_{\mathcal{M}}$($\Omega$), then there exists a rule $r'$ $\in$ $\mathcal{R}$ that strictly dominates $r$ and comp($r$,$r'$)=true, which means that $\forall$ $m$ $\in$ $\mathcal{M}$, $r'$[$m$] $\succeq$ $r$[$m$] and $\exists$ $m'$ $\in$ $\mathcal{M}$, $r'$[$m'$] $\succ$ $r$[$m'$]. Hence, we have $DegSim$($r^{\perp}$,$r'$) $<$ $DegSim$($r^{\perp}$,$r$). The latter inequality contradicts our hypothesis, since $r$ has the minimal degree of similarity with respect to $r^{\perp}$
\end{proof}

After identifying a representative rule $r$, the rules comparable and dominated by $r$ must be identified by comparing them to $r$. Na\"ively, $r$ must be compared to all rules in $\mathcal{R}$, yet we show in the following that we can even reduce the set of rules to be compared with $r$ into a subset of $\mathcal{R}$.

\begin{lemma}\label{l2}
Let $r$, $r'$, $r''$ $\in$ $\mathcal{R}$ with $r'$ $\in$ $Incomp$($r$).\\
If $r$ $\not\succ$ $r''$ then $r'$ $\not\succ$ $r''$
\end{lemma}

\begin{proof}
$r'$ $\in$ $Incomp$($r$) implies that $r$ $\succ$ $r'$. If $r$ $\not\succ$ $r''$ then there are two cases:
\begin{enumerate}
  \item Either $r$ $\equiv$ $r''$, then obviously $r$ $\equiv$ $r''$ $\succ$ $r'$
  \item Or $r$ $\not\equiv$ $r''$, then $\exists$ $m$ $\in$ $\mathcal{M}$ such that $r''$[$m$] $\succ$ $r$[$m$] $\succeq$ $r'$[$m$]
\end{enumerate}
Thus, in both cases $r'$ cannot dominate $r''$.
\end{proof}

Lemma \ref{l2} states that any rule $r'$ belonging to $Incomp$($r$) cannot dominate a rule which does not belong to $Incomp$($r$). In consequence, if $r'$ is a representative rule, then it is useless to compare it with the rules which are not dominated by $r$. The next lemma allows us to characterize the set of candidate rules that can be eliminated by $r'$.

\begin{lemma}\label{l3}
Let $r$, $r'$, $r''$ $\in$ $\mathcal{R}$ with $r'$ $\in$ $Incomp$($r$).\\
If $r'$ $\succ$ $r''$ and comp($r$, $r''$) = $false$  then $r''$ $\in$ $Incomp$($r$)
\end{lemma}

\begin{proof}
$r'$ $\in$ $Incomp$($r$) implies that $r$ $\succ$ $r'$. If $r'$ $\succ$ $r''$ then by the dominance transitivity $r'$ $\succ$ $r''$. Further, if $comp$($r$, $r''$) = $false$, then according to the definition \ref{comp} $r''$ $\in$ $Incomp$($r$).
\end{proof}

In what follows, we show how we can reduce the set of rules to be compared with an undominated rule.

\begin{definition}(undominated space)\label{def9}
Let $r$ be an undominated rule. If there exists a rule $r'$ which is not dominated by $r$ such that $r$ $\not\equiv$ $r'$, then there exists at least a measure $m$ $\in$ $\mathcal{M}$ such that $r'[m]$ $\succ$ $r[m]$. Since there exist $k$ measures in $\mathcal{M}$, then there are $k$ sets such that each one of them may contain rules not dominated by $r$. For each measure $m_{i}$ $\in$ $\mathcal{M}$, $i$=1,...,$k$, the corresponding set $s^{r}_{i}$ of rules which are not dominated by $r$ is defined as follows:


\begin{center}
$s^{r}_{i}$ = $\{$ $r'$ $\in$ $\mathcal{R}$ $\mid$ $r$ $\nsucc$ $r'$ and $r'$[$m_{i}$] $\succ$ $r$ [$m_{i}$]$\}$
\end{center}


These $k$ sets compose the undominated space of $r$, denoted $\mathcal{S}^{r}$=$\{$$s^{r}_{i}$$\}$, $i$=1,...,$k$.
\end{definition}

\begin{example} From our toy example presented in Table \ref{example}(b), for the undominated rule $r_{2}$, we have $s^{r_{2}}_{1}$ = $\emptyset$, $s^{r_{2}}_{2}$ = $\emptyset$ and $s^{r_{2}}_{3}$ = $\{$$r_{5}$$\}$. $s^{r_{2}}_{1}$ and $s^{r_{2}}_{2}$ are empty since there is no rule $r$ $\in$ $\mathcal{R}$ such that $r$[$m_{1}$] $\succ$ $r_{2}$[$m_{1}$] or $r$[$m_{2}$] $\succ$ $r_{2}$[$m_{2}$]. However, $s^{r_{2}}_{3}$ contains $r_{5}$ since $r_{5}$[$m_{3}$] $\succ$ $r_{2}$[$m_{3}$]. Following a similar reasoning, for the undominated rule $r_{5}$, we have $s^{r_{5}}_{1}$ = $\emptyset$, $s^{r_{5}}_{2}$ = $\{$$r_{1}$, $r_{2}$, $r_{3}$, $r_{12}$, $r_{13}$, $r_{14}$$\}$ and $s^{r_{5}}_{3}$ = $\emptyset$.
\end{example}

\begin{lemma}\label{l4}
Let $r$,$r'$ $\in$ $\mathcal{R}$ be two undominated rules and $s^{r}$ $\in$ $\mathcal{S}^{r}$. If $r'$ $\not\in$ $s^{r}$, then $\forall$ $r''$ $\in$ $s^{r}$, $r'$ $\not\succ$$r''$.
\end{lemma}

\begin{proof}
Given $r$, $r'$ $\in$ $\mathcal{R}$ two undominated rules and $s^{r}$ $\in$ $\mathcal{S}^{r}$ corresponding to a measure $m$ $\in$ $\mathcal{M}$. If $r'$ $\not\in$ $s^{r}$, then $r'$[$m$] $\nsucc$ $r$[$m$] which means that $r$[$m$] $\succeq$ $r'$[$m$] (1). Moreover, since $r''$ $\in$ $s^{r}$ then $r''$[$m$] $\succ$ $r$[$m$] (2). According to the dominance transitivity, (1) and (2) lead to $r''$[$m$] $\succ$ $r'$[$m$]. Hence, $r'$ $\not\succ$$r''$.
\end{proof}

\begin{lemma}\label{l5}
Let be $r$, $r'$ $\in$ $\mathcal{R}$ and $s^{r}$ $\in$ $\mathcal{S}^{r}$ such that $r$ is an undominated rule and $r'$ $\in$ $s^{r}$.
If $r'$ has the minimal degree of similarity with respect to $r^{\perp}$ among the rules in $s^{r}$, then $r'$ $\in$ $\mathcal{RR}$$_{\mathcal{M}}$($\Omega$).
\end{lemma}

\begin{proof}
Given $r$, $r'$ $\in$ $\mathcal{R}$ and $s^{r}$ $\in$ $\mathcal{S}^{r}$ such that $r'$ $\in$ $s^{r}$ and $r'$ has the minimal degree of similarity with $r^{\perp}$ among the rules in $s^{r}$. Suppose that $r'$ $\not\in$ $\mathcal{RR}$$_{\mathcal{M}}$($\Omega$), then it means that there exists a rule $r''$ $\in$ $\mathcal{R}$ such that $r''$$\succ$$r'$ and $comp$($r'$,$r''$)=$true$. According to lemma \ref{l4}, $r''$ must be in $s^{r}$ since any rule not belonging to $s^{r}$ cannot dominate $r'$. Moreover, $\forall$ $m$ $\in$ $\mathcal{M}$, $r''$[$m$] $\succeq$ $r'$[$m$] and $\exists$ $m'$ $\in$ $\mathcal{M}$, $r''$[$m'$] $\succ$ $r'$[$m'$]. Hence,  $DegSim$($r^{\perp}$,$r''$) $<$ $DegSim$($r^{\perp}$,$r'$) which contradicts our hypothesis since $r'$ has the minimal degree of similarity with $r^{\perp}$ in\ \ $s^{r}$.
\end{proof}

\subsection{Algorithm discovering the representative rules}

Based on the formalization, we propose the \textsc{RAR} algorithm allowing to discover representative rules. In \textsc{RAR}, we use the following variables for accumulating data during the execution of the algorithm:

\begin{itemize}
  \item[-] The variable $\mathcal{RR}$: is a variable initialized to empty set, it is used to keep track of the representative rules.\vspace*{-0.2cm}
  \item[-] The variable $Incomp$: is a variable that contains a subset of current candidate rules to be qualified as representative. This subset contains only rules which are incomparable with rules of $\mathcal{RR}$; it is initialized to empty set.\vspace*{-0.3cm}
  \item[-] The variable $S$: is a variable that contains all current sets covering the undominated space of all undominated rules; it is initialized to ${\mathcal{R}}$ since initially, all rules are considered undominated.
\end{itemize}

Informally, the algorithm works as follows:
\begin{itemize}
  \item[-] If the set of candidate rules $S$ and $Incomp$ are both empty, then the algorithm terminates and all representative rules are outputted through the variable $\mathcal{RR}$.
  \item[-] Otherwise, each rule $r$ in $\{S$$\cup$$Incomp\}$ might be a representative one. If $r$ has the minimal degree of similarity with the reference rule $r^{\perp}$, then $r$ is a representative rule and is added to $\mathcal{RR}$. Two cases have to be distinguished:
      \begin{enumerate}
        \item if $r$ belongs to incomparable set, then $r$ is no longer candidate and it is withdrawn from $Incomp$. After that, only the incomparable set is explored in order to delete rules which are comparable with $r$ and also dominated by $r$.
        \item otherwise (\emph{i.e.,} $r$ belongs to $S$), both the incomparable set set and the undominated space containing $r$ are explored. From the incomparable set, the rules comparable and dominated by $r$ will be removed. The undominated space containing $r$ is explored as follows:
            for each rule $r'$, in undominated space, is compared with $r$. Three cases have to be distinguished:
            \begin{enumerate}
              \item if $r'$ and $r$ are comparable and $r$ dominates $r'$, then $r'$ is no longer candidate and it is withdrawn from $S$.
              \item if $r'$ is incomparable with $r$, then $r'$ is still a candidate rule and it is added to the $Incomp$ set.
              \item otherwise, $r'$ is not dominated by $r$, \emph{i.e.,} $r'$ is still a candidate rule and it is added to the undominated subspace of $r$ according to definition \ref{def9}.
           \end{enumerate}
      Then, the undominated space containing $r$ is deleted from $S$ and the undominated space of $r$ is added to $S$. This process comes to an end when all candidates are handled.
      \end{enumerate}
\end{itemize}

\begin{algorithm}
\KwIn
{$\Omega$ = ($\mathcal{R}$, $\mathcal{M}$)}
\KwOut
{$\mathcal{RR}$: Representative rules}

\Begin
{
$\mathcal{RR}$ = $\emptyset$, $Incomp$ = $\emptyset$, $S$ = $\mathcal{R}$ \\

    \While{$S$ $\neq$ $\emptyset$ or $Incomp$ $\neq$ $\emptyset$}
        {
           $r^{*}$ a rule belonging to $S$ $\cup$ $Incomp$ having $min$($DegSim$($r$,$r^{\perp}$))\\

           add $r^{*}$ to $\mathcal{RR}$\\
           \ForEach{$r$ $\in$ $Incomp$}
            {
              \If{$r^{*}$ $\succ$ $r$ and comp($r^{*}$,$r$)}
              {
               remove $r$ from $Incomp$
              }
            }

            \If{$r^{*}$ $\in$ $Incomp$}
            {
              remove $r^{*}$ from $Incomp$
            }

            \Else
            {
              \ForEach{subspace $s$ $\in$ $S$ such that $r^{*}$ $\in$ $s$}
              {
                 \ForEach{$r$ $\in$ $s$}
                 {
                   \If{$r^{*}$ $\succ$ $r$}
                   {
                     \If{not(comp($r^{*}$, $r$))}
                     {
                      add $r$ to $Incomp$
                     }
                     remove $r$ from $S$
                   }
                   \Else
                   {
                     \ForAll {i such that $r$[$m_{i}$] $>$ $r^{*}$[$m_{i}$]}
                     {
                      add $r$ to the new subspace $s^{r^{*}}_{i}$
                     }
                   }
                 remove $s$ from $S$
                 }
              }
               add $\cup_{i}$$s^{r^{*}}_{i}$ to $S$
            }

        }

  \KwRet{$\mathcal{RR}$}
}

\caption{\textsc{RAR}}
  \label{Sky}
\end{algorithm}

\vspace*{0.4cm}

\section{Experimental study}\label{s4}

The aim of this experimental study is twofold. First, we show through extensive experiments that \textsc{RAR} provides interesting instance reduction compared to the initial set of rules. Second, we assess whether the number of measures has any uniform impact on the number of representative rules. These experiments were carried out on benchmark datasets taken from the UCI Machine Learning Repository. Table \ref{bases} summarizes the characteristics of these datasets. All the tests were performed on a 1.73 GHz Intel processor with Linux operating system and 2 GB of main memory.

\vspace{0.5cm}

\begin{table}[htpd]

 \begin{center}
    \begin{tabular}{|c|c c c|}
  \hline
       $Dataset$    &  $\sharp\ items$  & $\sharp\ transactions$   & $Avg.\ size$  \\
          &   &   & $of\ transactions$  \\
       \hline
    $Diabete$    & 75 & 3196 & 37 \\
    $Flare$    & 39 & 1389 & 10 \\
    $Iris$ & 119 & 8124 & 23 \\
    $Monks1$   & 19 & 124 & 7 \\
    $Monks2$   & 19 & 169 & 7 \\
    $Monks3$   & 19 & 122 & 7 \\
    $Nursery$  & 32 & 12960 & 9 \\
    $Zoo$      & 42 & 101 & 9 \\\hline
\end{tabular}
 \end{center}

\caption{Benchmark dataset characteristics.}
\label{bases}
\end{table}


\subsection{Reduction of number of rules}

\begin{table*}[!t]
  \centering
  \caption{Representative rules \emph{vs} undominated rules, TB rules and all rules}
    \begin{tabular}{|c|c|cc|cc|cc|cc|cc|cc|cc|}
    \hline
    \multicolumn{2}{|c|}{Datasets} & \multicolumn{2}{p{1.6cm}|}{$\{$Conf;Loev$\}$} & \multicolumn{2}{p{1.6cm}|}{$\{$Conf;Pearl$\}$} & \multicolumn{2}{p{1.8cm}|}{$\{$Conf;Recall$\}$} & \multicolumn{2}{p{1.8cm}|}{$\{$Conf;Zhang$\}$} & \multicolumn{2}{c|}{$\{$Conf;Pearl} & \multicolumn{2}{c|}{$\{$Conf;Loev} & \multicolumn{2}{c|}{$\{$Conf;Loev;Pearl} \\
    \multicolumn{2}{|c|}{(\emph{min$_{freq}$} \%)} & \multicolumn{2}{c|}{} & \multicolumn{2}{c|}{} & \multicolumn{2}{c|}{} & \multicolumn{2}{c|}{} & \multicolumn{2}{c|}{Recall$\}$} & \multicolumn{2}{c|}{Zhang$\}$} & \multicolumn{2}{c|}{Recall;Zhang$\}$} \\
    \hline
    \multirow{3}[0]{*}{Diabetes}& $\mathcal{RR}$ & \multicolumn{2}{r|}{5084} & \multicolumn{2}{r|}{619} & \multicolumn{2}{r|}{8512} & \multicolumn{2}{r|}{4931} & \multicolumn{2}{r|}{481} & \multicolumn{2}{r|}{1315} & \multicolumn{2}{r|}{1012} \\
    & $\mathcal{S}$$ky$-$\mathcal{R}$ & \multicolumn{2}{r|}{3411} & \multicolumn{2}{r|}{9} & \multicolumn{2}{r|}{6651} & \multicolumn{2}{r|}{2996} & \multicolumn{2}{r|}{9} & \multicolumn{2}{r|}{171} & \multicolumn{2}{r|}{171} \\
          & $\mathcal{TB}$-$\mathcal{R}$  & \multicolumn{2}{r|}{59314} & \multicolumn{2}{r|}{58124} & \multicolumn{2}{r|}{59206} & \multicolumn{2}{r|}{59309} & \multicolumn{2}{r|}{44813} & \multicolumn{2}{r|}{44602} & \multicolumn{2}{r|}{42126} \\
          (10,00) & $\mathcal{A}$-$\mathcal{R}$   & \multicolumn{2}{r|}{62132} & \multicolumn{2}{r|}{62132} & \multicolumn{2}{r|}{62132} & \multicolumn{2}{r|}{62132} & \multicolumn{2}{r|}{62132} & \multicolumn{2}{r|}{62132} & \multicolumn{2}{r|}{62132} \\\hline\hline
    \multirow{3}[0]{*}{Flare}& $\mathcal{RR}$ & \multicolumn{2}{r|}{6883} & \multicolumn{2}{r|}{502} & \multicolumn{2}{r|}{6993} & \multicolumn{2}{r|}{6817} & \multicolumn{2}{r|}{443} & \multicolumn{2}{r|}{269} & \multicolumn{2}{r|}{291} \\
    & $\mathcal{S}$$ky$-$\mathcal{R}$ & \multicolumn{2}{r|}{4975} & \multicolumn{2}{r|}{48} & \multicolumn{2}{r|}{4978} & \multicolumn{2}{r|}{4857} & \multicolumn{2}{r|}{48} & \multicolumn{2}{r|}{48} & \multicolumn{2}{r|}{48} \\
          & $\mathcal{TB}$-$\mathcal{R}$  & \multicolumn{2}{r|}{56163} & \multicolumn{2}{r|}{57101} & \multicolumn{2}{r|}{56451} & \multicolumn{2}{r|}{54524} & \multicolumn{2}{r|}{53197} & \multicolumn{2}{r|}{53116} & \multicolumn{2}{r|}{52819} \\
        (10,00)  & $\mathcal{A}$-$\mathcal{R}$   & \multicolumn{2}{r|}{57476} & \multicolumn{2}{r|}{57476} & \multicolumn{2}{r|}{57476} & \multicolumn{2}{r|}{57476} & \multicolumn{2}{r|}{57476} & \multicolumn{2}{r|}{57476} & \multicolumn{2}{r|}{57476} \\\hline\hline
    \multirow{3}[0]{*}{Iris}& $\mathcal{RR}$ & \multicolumn{2}{r|}{302} & \multicolumn{2}{r|}{265} & \multicolumn{2}{r|}{271} & \multicolumn{2}{r|}{262} & \multicolumn{2}{r|}{261} & \multicolumn{2}{r|}{264} & \multicolumn{2}{r|}{253} \\
    & $\mathcal{S}$$ky$-$\mathcal{R}$ & \multicolumn{2}{r|}{246} & \multicolumn{2}{r|}{246} & \multicolumn{2}{r|}{246} & \multicolumn{2}{r|}{246} & \multicolumn{2}{r|}{246} & \multicolumn{2}{r|}{246} & \multicolumn{2}{r|}{246} \\
          & $\mathcal{TB}$-$\mathcal{R}$   & \multicolumn{2}{r|}{440} & \multicolumn{2}{r|}{440} & \multicolumn{2}{r|}{440} & \multicolumn{2}{r|}{440} & \multicolumn{2}{r|}{440} & \multicolumn{2}{r|}{440} & \multicolumn{2}{r|}{440} \\
       (0,00)   & $\mathcal{A}$-$\mathcal{R}$   & \multicolumn{2}{r|}{440} & \multicolumn{2}{r|}{440} & \multicolumn{2}{r|}{440} & \multicolumn{2}{r|}{440} & \multicolumn{2}{r|}{440} & \multicolumn{2}{r|}{440} & \multicolumn{2}{r|}{440} \\\hline\hline
    \multirow{3}[0]{*}{Monks1}& $\mathcal{RR}$ & \multicolumn{2}{r|}{3883} & \multicolumn{2}{r|}{2106} & \multicolumn{2}{r|}{2891} & \multicolumn{2}{r|}{2797} & \multicolumn{2}{r|}{1003} & \multicolumn{2}{r|}{816} & \multicolumn{2}{r|}{694} \\
    & $\mathcal{S}$$ky$-$\mathcal{R}$ & \multicolumn{2}{r|}{768} & \multicolumn{2}{r|}{1} & \multicolumn{2}{r|}{788} & \multicolumn{2}{r|}{656} & \multicolumn{2}{r|}{1} & \multicolumn{2}{r|}{1} & \multicolumn{2}{r|}{1} \\
          & $\mathcal{TB}$-$\mathcal{R}$   & \multicolumn{2}{r|}{60417} & \multicolumn{2}{r|}{60692} & \multicolumn{2}{r|}{59418} & \multicolumn{2}{r|}{59452} & \multicolumn{2}{r|}{58904} & \multicolumn{2}{r|}{58811} & \multicolumn{2}{r|}{58327} \\
      (1,00)  & $\mathcal{A}$-$\mathcal{R}$  & \multicolumn{2}{r|}{62184} & \multicolumn{2}{r|}{62184} & \multicolumn{2}{r|}{62184} & \multicolumn{2}{r|}{62184} & \multicolumn{2}{r|}{62184} & \multicolumn{2}{r|}{62184} & \multicolumn{2}{r|}{62184} \\\hline\hline
    \multirow{3}[0]{*}{Monks2}& $\mathcal{RR}$ & \multicolumn{2}{r|}{414} & \multicolumn{2}{r|}{287} & \multicolumn{2}{r|}{503} & \multicolumn{2}{r|}{471} & \multicolumn{2}{r|}{215} & \multicolumn{2}{r|}{227} & \multicolumn{2}{r|}{223} \\
    & $\mathcal{S}$$ky$-$\mathcal{R}$ & \multicolumn{2}{r|}{279} & \multicolumn{2}{r|}{3} & \multicolumn{2}{r|}{215} & \multicolumn{2}{r|}{202} & \multicolumn{2}{r|}{3} & \multicolumn{2}{r|}{3} & \multicolumn{2}{r|}{3} \\
          & $\mathcal{TB}$-$\mathcal{R}$  & \multicolumn{2}{r|}{59611} & \multicolumn{2}{r|}{59702} & \multicolumn{2}{r|}{59568} & \multicolumn{2}{r|}{59544} & \multicolumn{2}{r|}{59103} & \multicolumn{2}{r|}{58917} & \multicolumn{2}{r|}{58662} \\
        (1,00)  & $\mathcal{A}$-$\mathcal{R}$   & \multicolumn{2}{r|}{59976} & \multicolumn{2}{r|}{59976} & \multicolumn{2}{r|}{59976} & \multicolumn{2}{r|}{59976} & \multicolumn{2}{r|}{59976} & \multicolumn{2}{r|}{59976} & \multicolumn{2}{r|}{59976} \\\hline\hline
    \multirow{3}[0]{*}{Monks3}& $\mathcal{RR}$ & \multicolumn{2}{r|}{3107} & \multicolumn{2}{r|}{773} & \multicolumn{2}{r|}{2094} & \multicolumn{2}{r|}{2362} & \multicolumn{2}{r|}{1266} & \multicolumn{2}{r|}{814} & \multicolumn{2}{r|}{458} \\
    & $\mathcal{S}$$ky$-$\mathcal{R}$ & \multicolumn{2}{r|}{1028} & \multicolumn{2}{r|}{2} & \multicolumn{2}{r|}{713} & \multicolumn{2}{r|}{781} & \multicolumn{2}{r|}{4} & \multicolumn{2}{r|}{2} & \multicolumn{2}{r|}{2} \\
          & $\mathcal{TB}$-$\mathcal{R}$   & \multicolumn{2}{r|}{58662} & \multicolumn{2}{r|}{58369} & \multicolumn{2}{r|}{57922} & \multicolumn{2}{r|}{58436} & \multicolumn{2}{r|}{57816} & \multicolumn{2}{r|}{57734} & \multicolumn{2}{r|}{56038} \\
       (1,00)  & $\mathcal{A}$-$\mathcal{R}$   & \multicolumn{2}{r|}{59304} & \multicolumn{2}{r|}{59304} & \multicolumn{2}{r|}{59304} & \multicolumn{2}{r|}{59304} & \multicolumn{2}{r|}{59304} & \multicolumn{2}{r|}{59304} & \multicolumn{2}{r|}{59304} \\\hline\hline
    \multirow{3}[0]{*}{Nursery}& $\mathcal{RR}$ & \multicolumn{2}{r|}{2883} & \multicolumn{2}{r|}{658} & \multicolumn{2}{r|}{1738} & \multicolumn{2}{r|}{1846} & \multicolumn{2}{r|}{573} & \multicolumn{2}{r|}{612} & \multicolumn{2}{r|}{554} \\
    & $\mathcal{S}$$ky$-$\mathcal{R}$ & \multicolumn{2}{r|}{497} & \multicolumn{2}{r|}{2} & \multicolumn{2}{r|}{304} & \multicolumn{2}{r|}{342} & \multicolumn{2}{r|}{8} & \multicolumn{2}{r|}{2} & \multicolumn{2}{r|}{2} \\
          & $\mathcal{TB}$-$\mathcal{R}$   & \multicolumn{2}{r|}{23872} & \multicolumn{2}{r|}{23901} & \multicolumn{2}{r|}{23875} & \multicolumn{2}{r|}{23417} & \multicolumn{2}{r|}{23176} & \multicolumn{2}{r|}{22806} & \multicolumn{2}{r|}{22139} \\
      (2,00)   & $\mathcal{A}$-$\mathcal{R}$   & \multicolumn{2}{r|}{25062} & \multicolumn{2}{r|}{25062} & \multicolumn{2}{r|}{25062} & \multicolumn{2}{r|}{25062} & \multicolumn{2}{r|}{25062} & \multicolumn{2}{r|}{25062} & \multicolumn{2}{r|}{25062} \\\hline\hline
    \multirow{3}[0]{*}{Zoo}& $\mathcal{RR}$ & \multicolumn{2}{r|}{1216} & \multicolumn{2}{r|}{493} & \multicolumn{2}{r|}{1161} & \multicolumn{2}{r|}{1124} & \multicolumn{2}{r|}{477} & \multicolumn{2}{r|}{462} & \multicolumn{2}{r|}{446} \\
    & $\mathcal{S}$$ky$-$\mathcal{R}$ & \multicolumn{2}{r|}{9784} & \multicolumn{2}{r|}{36} & \multicolumn{2}{r|}{9415} & \multicolumn{2}{r|}{9112} & \multicolumn{2}{r|}{36} & \multicolumn{2}{r|}{36} & \multicolumn{2}{r|}{36} \\
          & $\mathcal{TB}$-$\mathcal{R}$  & \multicolumn{2}{r|}{67991} & \multicolumn{2}{r|}{67305} & \multicolumn{2}{r|}{67872} & \multicolumn{2}{r|}{66146} & \multicolumn{2}{r|}{65328} & \multicolumn{2}{r|}{65116} & \multicolumn{2}{r|}{63926} \\
       (10,00)   & $\mathcal{A}$-$\mathcal{R}$   & \multicolumn{2}{r|}{71302} & \multicolumn{2}{r|}{71302} & \multicolumn{2}{r|}{71302} & \multicolumn{2}{r|}{71302} & \multicolumn{2}{r|}{71302} & \multicolumn{2}{r|}{71302} & \multicolumn{2}{r|}{71302} \\
    \hline
    \end{tabular}%
  \label{skynumber}%
\end{table*}%

\begin{table*}[!t]
  \centering
  \caption{Effectiveness of representative rules on UCI benchmarks}
    \begin{tabular}{|c|r|r|r|r|r|}
    \hline
    \multicolumn{1}{|c|}{\multirow{2}{*}{Measures}} & \multicolumn{1}{c|}{Minimal number } & \multicolumn{1}{c|}{Average number} & \multicolumn{1}{c|}{Maximal number} & \multicolumn{1}{c|}{Average number} & \multicolumn{1}{c|}{Average gain} \\
    \multicolumn{1}{|c|}{} & \multicolumn{1}{c|}{of $\mathcal{RR}$} & \multicolumn{1}{c|}{of $\mathcal{RR}$} & \multicolumn{1}{c|}{of $\mathcal{RR}$} & \multicolumn{1}{c|}{of $\mathcal{TB}$-$\mathcal{R}$} & \multicolumn{1}{c|}{of $\mathcal{RR}$} \\\hline
    $\{$Conf;Loev$\}$ &   302    &  2971,50     &   6883    &   48308,75    & 16,73 \\
    $\{$Conf;Pearl$\}$&    265   &   712,87    &   2106    &  40908,12     &  17,27\\
    $\{$Conf;Recall$\}$ &   271    &   2913,75    &   8512    &   48094,00    & 1,64 \\
    $\{$Conf;Zhang$\}$&   262    &   3020,37    &  6817     &   47658,50    & 29,31 \\
    $\{$Conf;Loev;Recall$\}$&    261   &  589,87     &   1266    &  45347,12     & 177,39 \\
    $\{$Conf;Pearl;Zhang$\}$ &   227    &   597,37    &    1315   &   45192,75    &  37,24\\
    $\{$Conf;Loev;Pearl;Recall;Zhang$\}$&  223     &   491,37    &   1012    &  44309,62     & 86,20 \\
    \hline
    \end{tabular}%
  \label{skyavrg}%
\end{table*}%

In the following, we show the ability of our approach to considerably reduce the oversized sets of rules generated from our experimental datasets. Our experiments batch aims to compare our approach to another one based on thresholds. For this purpose, we assign for each measure $m$ $\in$ $\mathcal{M}$, a threshold $\varepsilon_{m}$ such that $\varepsilon_{m}$ is the minimum value of the representative rules with respect to $m$, \emph{i.e.,} $\varepsilon_{m}$ $=$ $min$$\{$$r$[$m$] $|$ $r$ $\in$ $\mathcal{RR}$$_{\mathcal{M}}$($\Omega$). This ensures that all representative rules will be generated from the algorithm based on thresholds. For instance, in our running example (\textit{c.f.,} Table \ref{example}(b)), $\varepsilon_{freq}$ = 0.10, $\varepsilon_{conf}$ = 0.16 and $\varepsilon_{pearl}$ = 0.00. The set of resulting rules is called the threshold-based rules denoted by $\mathcal{TB}$-rules. These experiments have the benefit of quantifying the reduction of rules brought by \textsc{RAR} in the case where a user is able to perfectly specify thresholds for mining association rules algorithm based on thresholds. Hence, we compare the number of representative rules with respect to that of $\mathcal{TB}$-rules and the total number of association rules (denoted $\mathcal{A}$-$\mathcal{R}$). We considered a number of combinations of measures: Confidence \cite{Agra93}, Recall \cite{LavracFZ99}, Pearl \cite{Pearl88}, Loevinger \cite{Loev47}, Zhang \cite{Zhang01}.

For each set of measures, Table \ref{skynumber} compares the size of representative rules $\mathcal{RR}$ versus that of undominated rules (denoted by $\mathcal{S}$$ky$-$\mathcal{R}$), that of $\mathcal{TB}$-rules and that of all association rules. A major result is that the gain of representative rules is always important. Indeed the set of representative rules is very small compared to $\mathcal{TB}$-rules. This shows that even though using the the optimal threshold, the dimensionality problem of the huge number of rules remains. Table \ref{skyavrg} summarizes this result by sketching, for each set of measures, the minimal$/$average$/$maximal number of representative rules, the average number of $\mathcal{TB}$-rules and the average gain of representative rules versus the $\mathcal{TB}$-rules. The average gain rate is measured as follows:  $\frac{size\ of\ \mathcal{TB}-rules}{size\ of\ \mathcal{RR}}$.\\
A second observation is that the number of undominated rules is often extremely low, it even reaches less then 10. The explanation is that an undominated rule eliminates every rule it dominates, even if they are not comparable \emph{i.e} they are not semantically related.


\subsection{Impact of measure variation on the number of rules}

In what follows, we put the focus on the evolution of the representative rules cardinalities with respect to measure variation. Table \ref{skynumber} shows the effect of variation of $\mathcal{M}$ on representative rules, undominated rules, $\mathcal{TB}$-rules and all rules. We can notice that the number of all rules is obviously constant. In contrast, the number of $\mathcal{TB}$-rules is sensitive to the variation of cardinality of $\mathcal{M}$. Indeed, by adding each time a measure to $\mathcal{M}$, the number of $\mathcal{TB}$-rules decreases. However, the number of representative rules may decrease or increase. The decrease can be explained by the fact that an association rule can be undominated with respect to a set of measure $M_{1}$ and dominated with respect to $M_{2}$, such that $M_{1}$ $\subset$ $M_{2}$. For example, if two rules $r$ and $r'$ are equivalent and undominated with respect to $M_{1}$, there is a possibility that one of them dominates the other by considering one more measure. On the other hand, the increase can be explained by the fact that an association rule can be dominated with respect to $M_{1}$ and undominated with respect to $M_{2}$. For example, consider a rule $r$ which dominates another $r'$ with respect to $M_{1}$, by adding a measure $m$ to $M_{1}$, such that $r'$[$m$] $\succ$ $r$[$m$], then $r'$ is no longer dominated by $r$.

\section{Conclusion}\label{s5}

In this paper, we introduced an approach that addresses the problem of rule selection. This approach is not hindered by the abundance of measures which has been the issue of several works. These works have been devoted to measure selection in order to find one best measure, whereas the real issue lies in selecting rules to help with decision making. We proposed \textsc{RAR}, an algorithm to perform this task based on the dominance and comparability relationships. When using our algorithm, the user does not have to worry neither about the heterogeneity of measures nor about specifying thresholds. On the other hand, experimental results carried out on benchmark datasets showed important profits in terms of compactness of the representative rules.

An important direction for future work consists on setting up an approach aimed at discovering representative rules during the phase of the extraction rules which will improve the performance of the \textsc{RAR} algorithm. Another important task is to rank representative rules in order to answer to a personalized user query. Indeed, the user may ask to select top-k rules among representative rules. This selection cannot be performed unless a ranking is carried. Hence, it would be useful to set up a ranking process for the representative rules.

%

%


\bibliographystyle{abbrv}
\bibliography{KDD2}  
%
%

\end{document}